\def\final{1}

\documentclass[11pt]{article}
	
\usepackage{amsthm, amsmath, amssymb}
\usepackage{comment}
\usepackage{enumerate}
\usepackage{enumitem}
\usepackage{framed}
\usepackage{verbatim}
\usepackage[margin=1.1in]{geometry}
\usepackage{microtype}
\usepackage[T1]{fontenc}
\usepackage{kpfonts}
	\DeclareMathAlphabet{\mathsf}{OT1}{cmss}{m}{n}
	\SetMathAlphabet{\mathsf}{bold}{OT1}{cmss}{bx}{n}
	\DeclareMathAlphabet{\mathtt}{OT1}{cmtt}{m}{n}
	\SetMathAlphabet{\mathtt}{bold}{OT1}{cmtt}{bx}{n}

\usepackage{algorithm, algorithmic}
\usepackage{color}
	\definecolor{DarkGreen}{rgb}{0.15,0.5,0.15}
	\definecolor{DarkRed}{rgb}{0.6,0.2,0.2}
	\definecolor{DarkBlue}{rgb}{0.15,0.15,0.55}
	\definecolor{DarkPurple}{rgb}{0.4,0.2,0.4}

\usepackage[pdftex]{hyperref}
\hypersetup{
	linktocpage=true,
	colorlinks=true,			% false: boxed links; true: colored links
	linkcolor=DarkBlue,	% color of internal links
	citecolor=DarkBlue,	% color of links to bibliography
	urlcolor=DarkBlue,		% color of external links
	}
	
\ifnum\final=0
\newcommand{\mynote}[2]{\marginpar{\color{#1}\tiny #2}}
\newcommand{\mybignote}[2]{{\color{#1} $\langle \langle$ #2$\rangle \rangle$}}
\else
\newcommand{\mynote}[2]{}
\newcommand{\mybignote}[2]{}
\fi
\newcommand{\jnote}[1]{\mynote{DarkRed}{Jon: {#1}}}

\newcommand{\tnote}[1]{\mynote{DarkBlue}{Thomas: {#1}}}

\newcommand{\ex}[2]{{\ifx&#1& \mathbb{E} \else \underset{#1}{\mathbb{E}} \fi \left[#2\right]}}
\newcommand{\pr}[2]{{\ifx&#1& \mathbb{P} \else \underset{#1}{\mathbb{P}} \fi \left[#2\right]}}
\newcommand{\var}[2]{{\ifx&#1& \mathsf{Var} \else \underset{#1}{\mathsf{Var}} \fi \left[#2\right]}}

\newcommand{\nope}[1]{}

\newcommand{\zo}{\{0,1\}}

\newcommand{\from}{:}

\renewcommand{\epsilon}{\varepsilon}
\newcommand{\eps}{\varepsilon}

\DeclareMathOperator*{\argmax}{arg\,max}

\newcommand{\E}{\mathbb{E}}
\newcommand{\N}{\mathbb{N}}

\newcommand{\cP}{\mathcal{P}}

\newtheorem{lem}{Lemma}
\newtheorem{defn}[lem]{Definition}
\newtheorem{cor}[lem]{Corollary}
\newtheorem{prop}[lem]{Proposition}
\newtheorem{thm}[lem]{Theorem}
\newtheorem{conj}[lem]{Conjecture}

\newcommand{\BetaD}{\mathsf{Beta}}
\newcommand{\BetaF}{\mathsf{B}}

\title{Tight Lower Bounds for Differentially Private Selection}
\author{Thomas Steinke\thanks{IBM Research -- Almaden. \href{mailto:topk@thomas-steinke.net}{\texttt{topk@thomas-steinke.net}}.} \and Jonathan Ullman\thanks{Northeastern University, College of Computer and Information Science.  \href{mailto:jullman@ccs.neu.edu}{\texttt{jullman@ccs.neu.edu}}}}

\begin{document}
\pagenumbering{gobble}
\maketitle

\begin{abstract}
A pervasive task in the differential privacy literature is to select the $k$ items of ``highest quality'' out of a set of $d$ items, where the quality of each item depends on a sensitive dataset that must be protected.  Variants of this task arise naturally in fundamental problems like feature selection and hypothesis testing, and also as subroutines for many sophisticated differentially private algorithms.

The standard approaches to these tasks---repeated use of the exponential mechanism or the sparse vector technique---approximately solve this problem given a dataset of $n = O(\sqrt{k}\log d)$ samples.  \jnote{EM DASHES DO NOT HAVE SPACES AROUND THEM!!!}  We provide a tight lower bound for some very simple variants of the private selection problem.  Our lower bound shows that a sample of size $n = \Omega(\sqrt{k} \log d)$ is required even to achieve a very minimal accuracy guarantee.

Our results are based on an extension of the fingerprinting method to sparse selection problems.  Previously, the fingerprinting method has been used to provide tight lower bounds for answering an entire set of $d$ queries, but often only some much smaller set of $k$ queries are relevant.  Our extension allows us to prove lower bounds that depend on both the number of relevant queries and the total number of queries.

\end{abstract}

\vfill
\newpage

\pagenumbering{arabic}
\section{Introduction} \label{sec:intro}
This work studies lower bounds on the sample complexity of differentially private selection problems. Informally, a selection problem consists of a large number of items each with a corresponding value and the task is to select a small subset of those items with large values.  In a private selection problem, the values of the items depend on a dataset of sensitive information that must be protected.

%In a private selection problem, the values of the items depend on the sensitive data of several individuals, and thus we want our selection procedure to ensure \emph{differential privacy} \cite{DworkMNS06}, which guarantees that no individual's data has a significant effect on the items selected.

Selection problems appear in many natural statistical problems, including private multiple hypothesis testing \cite{DworkSZ15}, sparse linear regression~\cite{SmithT13, TalwarTZ15}, finding frequent itemsets~\cite{BhaskarLST10}, and as subroutines in algorithms for answering exponentially many statistical queries~\cite{BlumLR08, RothR10, HardtR10, GuptaRU12, Ullman15}, approximation algorithms~\cite{GuptaLMRT10}, and for establishing the generalization properties of differentially private algorithms~\cite{BassilyNSSSU16}. Selection problems appear in many different guises. As we are proving lower bounds, we consider the simplest possible form of selection problems.

More specifically, we consider the following simple selection problem motivated by applications in feature selection and hypothesis testing. There is an unknown probability distribution $\cP$ over $\zo^{d}$ with mean $p := \ex{}{\cP} \in [0,1]^d$, and our goal is to identify a set of coordinates whose mean is large---that is, a set $S \subset [d]$ of size $k \ll d$, such that $p^j$ is large for all $j \in S$.  To do this, we obtain $n$ independent samples $X_1, \cdots, X_n \in \{0,1\}^d$ from $\cP$. However, each $X_i$ corresponds to the private data of an individual.\footnote{For clarity, we use superscripts to denote the index of a column or item and subscripts to denote the index of a row or individual.} To protect this data, our procedure for selecting $S$ using the data $X_1, \cdots, X_n$ should satisfy \emph{differential privacy} \cite{DworkMNS06}, which is a strong notion of privacy requiring that no individual sample $X_i$ has a significant influence on the set of coordinates $S$ that we select.

For example, suppose $\mathcal{P}$ represents a population of patients suffering from some illness and each coordinate represents the presence of absence of a certain genetic trait. It would be useful for medical researchers to identify genetic traits that are unusually common in this population, but it is also essential not to reveal any individual's genetic information.  Thus the researchers would like to obtain genetic data $X_1,\dots,X_n$ from $n$ random members of this population and run a differentially private selection algorithm on this dataset. 

Without privacy, it is necessary and sufficient to draw $n \gtrsim \log d$ samples from $\cP$, and compute $\overline{X} = \frac1n \sum_i X_i$.  This ensures that $\| \overline{X} - p \|_\infty$ is small with high probability,\footnote{More precisely, if $n \geq \frac{\log(2d/\beta)}{2 \alpha^2}$, then $\pr{}{\|\overline{X} - p\|_\infty \leq \alpha}\geq1-\beta.$  Since we are proving negative results, we focus on the low-accuracy regime of $\alpha, \beta = \Omega(1)$, where $n = \Theta(\log d)$ samples are both necessary and sufficient.} so large coordinates of $\overline{X}$ correspond to large coordinates of $p$.  We can ensure differential privacy by adding carefully calibrated noise to the empirical mean $\overline{X}$ to obtain a noisy empirical mean $\tilde{X}$~\cite{DinurN03, DworkN04, BlumDMN05, DworkMNS06}.  Unfortunately, there are strong lower bounds showing that, unless $n \gtrsim \sqrt{d}$, there is no differentially private algorithm whose output $\tilde{X}$ gives a useful approximation to the population mean $p$~\cite{BunUV14, SteinkeU15b, DworkSSUV15}.

We can avoid this $\sqrt{d}$ lower bound if we only want to identify the $k$ approximately largest coordinates of $p$, rather than approximating all $d$ values.  Specifically, we can use the \emph{exponential mechanism}~\cite{McSherryT07} to identify an approximate largest coordinate of $p$, and then repeat on the other coordinates.  This algorithm provides non-trivial error using just $n \gtrsim \sqrt{k} \log d $ samples. This sample complexity is also achieved by the sparse vector algorithm~\cite{DworkNRRV09} (see~\cite[\S 3.6]{DworkR14} for a textbook treatment) and report noisy max \cite[\S 3.3]{DworkR14}.

Our first result shows that this sample-complexity is essentially the best possible for the approximate top-$k$ selection problem, even if $\cP$ is a product distribution.
\begin{thm}[Informal version of Corollary \ref{cor:topklb}] \label{thm:intro1}
Fix $n,d,k \in \N$ with $k \ll d$.  Let $M$ be a differentially private algorithm that takes a dataset $X \in (\zo^d)^n$ of $n$ samples, and outputs an indicator vector $M(X) \in \zo^d$ such that $\|M(X)\|_1 = k$.  Suppose that for every product distribution $\cP$ over $\zo^d$,
\begin{equation}
\ex{X \leftarrow \cP^n \atop M}{\sum_{j \in [d] ~:~ M(X)^j = 1} p^j} \geq \max_{ t \in \zo^d \atop \|t\|_1 = k} \sum_{j \in [d] ~:~ t^j = 1} p^j - \frac{k}{10},
\label{eqn:pop-acc}\end{equation}
where $p = \ex{}{\cP}$.  Then $n = \Omega(\sqrt{k} \log d)$.
\end{thm}
Observe that our lower bound applies whenever the error is at most $k/10$, which is just slightly smaller than the trivial error of $k$ that can be obtained by selecting the first $k$ coordinates.

\vspace{-2mm}
\paragraph{Scaling with the Privacy and Accuracy Parameters.}  For simplicity, we suppress the dependence on the privacy and accuracy parameters in Theorem \ref{thm:intro1}. We assume constant error $\frac{1}{10}$ per selected coordinate, and our lower bound applies to algorithms satisfying $(1,1/nd)$-differential privacy. Generic reductions can be used to give the appropriate dependence on these parameters in many cases (see e.g.~\cite{BunUV14,SteinkeU15b}).  %Specifically, Theorem \ref{thm:intro1} implies, for all $\alpha \in (0,0.1)$, a lower bound of $n = \Omega(\frac{1}{\alpha} \sqrt{k} \log d)$ for any $(1, 1/nd)$-differentially private algorithm guaranteeing error $\alpha$ per coordinate. More involved reductions are known \cite{SteinkeU15b} that can introduce a factor of approximately $\frac{1}{\eps}\sqrt{\log(1/\delta)}$ into the lower bound for $(\eps,\delta)$-differential privacy. However these reductions are not directly applicable to the form of the lower bound we state here.

\subsubsection*{Empirical Error vs.~Population Error}
In Theorem \ref{thm:intro1}, accuracy was defined with respect to the population mean $p = \ex{}{\cP}$.  This statistical framework is motivated by the fact that we are interested in finding underlying patterns in the population, rather than random empirical deviations.

We could equally well define accuracy with respect to the empirical mean $\overline{X} = \frac1n \sum_i X_i$.  Since $\mathbb{E}[\| \overline{X} - p \|_{\infty}] \leq \sqrt{\frac{\log(2d)}{ 2n}}$ and we are interested in settings where $n \gtrsim \log d$, these settings are equivalent.\footnote{We need $n \gtrsim \log d$ even in the non-private statistical setting to have meaningful statistical accuracy. In the absence of privacy constraints, the empirical accuracy guarantee can be satisfied for every $n$.}  In particular, we can replace the accuracy condition \eqref{eqn:pop-acc} in Theorem \ref{thm:intro1} with %also show that every differentially private algorithm $M \from (\zo^d)^n \to \zo^d$ such that for every $X \in (\zo^d)^n$,
\begin{equation}
\ex{M}{\sum_{j \in [d] ~:~ M(X)^j = 1} \overline{X}^j} \geq \max_{ t \in \zo^d \atop \|t\|_1 = k} \sum_{j \in [d] ~:~ t^j = 1} \overline{X}^j - \frac{k}{20}.
\end{equation}
%requires a dataset of size $n = \Omega(\sqrt{k} \log(d))$.\footnote{In the absence of privacy constraints, the empirical accuracy guarantee can be satisfied for every $n$.}

This empirical variant of the problem was first studied in a very recent work by Bafna and Ullman~\cite{BafnaU17}.  They proved an optimal lower bound for the empirical variant of the problem in the regime where the error is very small. Specifically they show that, if the empirical error is $\ll k \sqrt{\log(d)/n}$ (i.e.~a constant factor smaller than the sampling error), then a dataset of size $n = \Omega(k \log d)$ is necessary.\footnote{The results of \cite{BafnaU17} actually use a slightly stronger accuracy requirement, which requires that for every $j \in [d]$, if $M(X)^j = 1$ then $\overline{X}^{(k)} - \overline{X}^j \ll \sqrt{\log(d)/n}$ where $\overline{X}^{(k)}$ is the $k$-th largest entry of $\overline{X}$.  This technical distinction is not crucial for this high-level discussion.}\tnote{I think the BU17 accuracy def is stronger, it's $\ell_\infty$ versus $\ell_1$.} \jnote{Yes, but it's technically incomparable because it only requires everything in $S$ to be larger than $P^{(k)}$.  Ours is proportional to $\sum_{j = 1}^{k} P^{(j)}$} However, their results do not give any lower bound for larger error, or for the statistical problem of approximating the largest entries of the population mean $p$.  Indeed, their lower bounds hold even for uniformly random datasets $X$.  For these datasets we can easily achieve empirical error $k\sqrt{2 \log(d)/n}$ and, since $p = (\frac12,\dots,\frac12)$ is fixed, we can achieve population error $0$.  Thus, our lower bounds for large error are qualitatively different from the lower bounds of \cite{BafnaU17} for small error.

\subsubsection*{Application to Multiple Hypothesis Testing}
We can prove an analogous lower bound for a related problem where we do not have a fixed number of coordinates $k$ that we want to select, but instead we want to distinguish coordinates of $p$ that are larger than some threshold $\tau$ from those that are smaller than some strictly lower threshold $\tau' < \tau$.  This problem is a special case of multiple hypothesis testing in statistics.  Without privacy it can be solved using just $n = O_{\tau, \tau'}(\log d)$ samples.

As before, we can use the exponential mechanism or the sparse vector technique to obtain a private algorithm for this problem.  The algorithm uses $n = O_{\tau, \tau'}(\sqrt{k} \log d)$ samples, where $k$ is an upper bound on the number of coordinates of $p$ that are above the threshold $\tau'$.\footnote{We assume that the upper bound $k$ is specified as part of the problem input. If $k$ is not specified, the problem and the accuracy guarantee can be formulated differently, but this is not relevant for the current high-level discussion.}

Our second result shows that this sample complexity is essentially optimal for the multiple hypothesis testing problem, even if $\cP$ is a product distribution.
\begin{thm}[Informal version of Corollary~\ref{cor:mhtlb}] \label{thm:intro2}
Fix $n,d,k \in \N$ with $k \ll d$.  There exist absolute constants $\tau, \tau', \rho \in (0,1)$, $\tau' < \tau$ such that the following holds.  Let $M$ be a differentially private algorithm that takes a dataset $X \in (\zo^d)^n$ of $n$ samples, and outputs an indicator vector $M(X) \in \zo^d$.  Suppose that for every product distribution $\cP$ over $\zo^d$ such that $| \{ j : p^j \geq \tau \}| \leq k$,
\begin{enumerate}
\item $p^j \leq \tau' \Longrightarrow \pr{X\leftarrow\cP^n, M}{M(X)^j = 1} \leq \rho k /d$,
\item $p^j \geq \tau \Longrightarrow \pr{X\leftarrow\cP^n, M}{M(X)^j = 1} \geq 1-\rho,$
\end{enumerate}
where $p = \ex{}{\cP}$.  Then $n = \Omega(\sqrt{k} \log d)$.
\end{thm}
As before, we remark that the fact that $\tau' = \tau - \Omega(1)$ makes our lower bound stronger.  Also, note that we allow the probability of a false positive ($p^j \leq \tau'$ but $M(X)^j = 1$) to be as large as $\rho k / d$, which means that in expectation there can be as many as $\Omega(k)$ of these false positives.  In contrast there are only $k$ true positives ($p^j \geq \tau$ and $M(X)^j = 1$), so our lower bound applies even to algorithms for which the number of false positives is a constant fraction of the number of true positives.  The accuracy condition in Theorem \ref{thm:intro2} is closely related to the \emph{false discovery rate}, which is a widely used statistical criterion introduced in the influential work of Benjamini and Hochberg~\cite{BenjaminiH79}.  A recent work by Dwork, Su, and Zhang~\cite{DworkSZ15} introduced the problem of privately controlling the false discovery rate.

\subsection{Techniques} \label{sec:intro:techniques}

Our techniques build on the recent line of work proving lower bounds in differential privacy and related problems via either fingerprinting codes or techniques inspired by fingerprinting codes~\cite{BunUV14, HardtU14, SteinkeU15a, SteinkeU15b, DworkSSUV15, BunSU17}. Our results follow from the following very general lower bound that refines and generalizes several of the results from those works.
\begin{thm} [Main Lower Bound] \label{thm:general}
Let $\beta,\gamma,k>0$ and $n,d \in \mathbb{N}$ be a fixed set of parameters.  Let $P^1, \cdots, P^d$ be independent draws from $\BetaD(\beta,\beta)$ and let $X \in (\{0,1\}^d)^n$ be a random dataset such that every $X_i^j$ is independent (conditioned on $P$) and $\E[X_i^j]=P^j$ for every $i \in [n]$ and $j \in [d]$.

Let $M : (\{0,1\}^d)^n \to [-1,1]^d$ be a $(1,\beta\gamma k / n d)$-differentially private algorithm and assume that $M$ satisfies the condition $\ex{P,X,M}{\|M(X)\|_2^2} = k$ and the accuracy condition
\begin{equation*}
\ex{P,X,M}{\sum_{j \in [d]} M(X)^j \cdot \left(P^j - \frac12\right)} \geq \gamma k.
\end{equation*}
Then $n \geq \gamma \beta \sqrt{k}.$%\footnote{For intuition, think of $\beta$ as a growing function ($\log d$ in our setting) that controls the ``sparsity'' --- that is, as $\beta \to \infty$, the population mean $P$ becomes more concentrated around its expectation $(\frac12,\frac12,\cdots,\frac12)$. Also $\gamma$ should be viewed as a small constant. More precisely, $\gamma$ is related to the accuracy parameter --- $M$ should ``detect'' when $P^j$ deviates from $\frac12$ by more than $\gamma$ and, hence, the more accurate $M$ is, the smaller we can take $\gamma$ (and $\beta$ should grow as $1/\gamma^2$).}
\end{thm}

% \jnote{Can we maybe give simpler version of this statement in the intro?  For example maybe with $\gamma, \beta = \Omega(1)$ and $\Delta \leq d$?}

%Note that it is necessary for our lower bound that the population mean $P$ itself be random. If the population mean is fixed, then simply outputting an appropriate function of this fixed value gives an algorithm which has perfect accuracy and is trivially differentially private. Thus the population mean must be uncertain or, as above, we must insist that $M$ work for all product distributions $\cP$.\tnote{Added explanatory remark and interpretation footnote.} \jnote{I am concerned that the footnote is a bit too technical, especially the parts about $\gamma$?  Also it's a bit redundant given the paragraph about the beta dist and the paragraph explaining how to go from Thm 3 to a top-$k$ LB.} \jnote{I also shortened the paragraph a bit.  It's kinda obvious that $P$ must be random since $P$ is the ``input'' to the selection problems we defined above.  But it doesn't hurt to clarify.  I still don't love how it's worded.}

In our applications, $\gamma=\Omega(1)$ is a small constant, whereas $\beta=\omega(1)$ is large, namely $\beta=\Theta(\log(d/k))$ for the top-$k$ lower bound. For the purposes of this introduction, it suffices to know that $\BetaD(\alpha, \beta)$ is a family of probability distributions over $[0,1]$ with mean $\frac{\alpha}{\alpha + \beta}$.  For simplicity, we restrict our attention to the symmetric case where $\alpha = \beta$. The distribution $\BetaD(1,1)$ is the uniform distribution on $[0,1]$ and the distribution becomes more concentrated around $1/2$ as $\beta \to \infty$, specifically the variance of $\BetaD(\beta, \beta)$ is $\Theta(\frac{1}{\beta})$.  The necessary technical details about the beta distribution are in Section~\ref{sec:prelims:beta}. 

Observe that in our lower bound the population mean $P$ is itself random.  If the population mean were fixed, then we could obtain a private algorithm with perfect accuracy by ignoring the sample and outputting a fixed function of $P$.  Thus to obtain lower bounds then we must assume that the distribution $\cP$ is chosen randomly and that $M$ is accurate for these distributions $\cP$.

We now describe informally how Theorem \ref{thm:general} implies Theorem \ref{thm:intro1}.  First, observe that any algorithm for approximate top-$k$ selection by definition satisfies $\ex{}{\|M(X)\|_2^2} = k$, since it outputs an indicator vector with exactly $k$ non-zero coordinates.  By Theorem \ref{thm:general}, to prove an $n = \Omega(\sqrt{k} \log d)$ lower bound, it suffices to show that for some $\beta = \Omega(\log d)$, if $M$ solves the approximate top-$k$ problem, then $\ex{}{\sum_j M(X)^j \cdot (P^j - \frac12)} = \Omega(k)$.  By the accuracy assumption \eqref{eqn:pop-acc}, it suffices to show
\begin{equation}
\ex{P}{\max_{ t \in \zo^d \atop \|t\|_1 = k} \sum_{j : t^j = 1} \left(P^j - \frac12\right) - \frac{k}{10} } \geq \Omega( k ) .
\label{eqn:anticoncentration}
\end{equation} 
This is simply a property of the beta distribution and our choice of $\beta$.  We give a simple anti-concentration result for beta distributions showing that the required bound \eqref{eqn:anticoncentration} holds for some choice of $\beta = \Omega(\log d)$.

We remark that previous fingerprinting-based lower bounds in differential privacy \cite{BunUV14,BassilyST14,DworkTTZ14,SteinkeU15b,DworkSSUV15,TalwarTZ15} essentially correspond to setting $\beta=O(1)$ in Theorem \ref{thm:general}. %, whence they cannot obtain stronger bounds by taking $\beta$ larger. 
Thus the key novelty of our result is that we obtain stronger lower bounds by setting $\beta=\omega(1)$.
\subsubsection*{Overview of the Analysis}

We will sketch the argument for our lower bound in the case of approximate top-$k$ selection.  Inspired by prior lower bounds~\cite{BunUV14, SteinkeU15b, DworkSSUV15, BafnaU17} we consider the quantity 
$$
Z := \sum_{i \in [n]} \langle M(X), (X_i - P) \rangle = n \cdot \left(\sum_{j : M(X)^j = 1} \overline{X}^j - P^j\right).
$$
We then use the privacy and accuracy assumptions to establish conflicting upper and lower bounds on the quantity $\ex{}{Z}$.  Combining the two bounds yields the result.  

Firstly, we use the differential privacy of $M$ to get an upper bound on $\ex{}{Z}$.  Specifically, for any $i \in [n]$, $M(X)$ should have approximately the same distribution as $M(X_{\sim i})$, where $X_{\sim i}$ is the dataset we obtain by replacing $X_i$ with an independent sample from $\cP$.  However, $X_i$ and $M(X_{\sim i})$ are independent (conditioned on $P$) and, therefore, $\ex{}{\langle M(X_{\sim i}), X_i - P \rangle}=0$. By differential privacy, $\ex{}{\langle M(X), X_i - P \rangle} \approx \ex{}{\langle M(X_{\sim i}), X_i - P \rangle} = 0$. More precisely, we obtain $\ex{}{\langle M(X), X_i - P \rangle} \leq O(\sqrt{k})$ and, thus, $\mathbb{E}[Z] \leq O(n \sqrt{k})$ (Lemma \ref{lem:upper}).

Secondly, if $M(X)$ solves the approximate top-$k$ selection problem, then $\mathbb{E}[Z]$ must be large (Lemma \ref{lem:lower}).  This is the technical heart of our result and requires extending the analysis of fingerprinting codes. We give some imprecise intuition for why we should expect $\ex{}{Z} \geq \Omega(k \beta)$.

The beta distribution has the following ``conjugate prior'' property. Suppose we sample $P \leftarrow \BetaD(\beta, \beta)$, independently sample $Y_1,\dots,Y_n \in \zo$ with mean $P$, and let $\overline{Y} = \frac{1}{n} \sum_{i} Y_i$. Then the conditional distribution of $P$ given $\overline{Y}$ is
$$(P \mid \overline{Y}=\overline{y}) \sim \BetaD\left( \beta + n\overline{y}, \beta + n(1-\overline{y})\right),~~~\textrm{so that}~~~\ex{}{P \mid \overline{Y}=\overline{y}} = \frac{\beta + n\overline{y}}{2\beta + n}.$$
Thus, if $\overline{y} \geq \frac12 + \Omega(1)$, then $$\ex{}{\overline{Y} - P \mid \overline{Y}=\overline{y}} = \frac{(2\overline{y}-1)\beta}{2\beta+n} = \Omega(\beta/n).$$  
We connect this back to $Z$ by observing that, if $M$ accurately solves the approximate top-$k$ selection problem, then it will identify a set of $k$ coordinates such that $\overline{X}^j = \frac12 + \Omega(1)$ on average over the selected indices $j$.  Applying this analysis and summing over the $k$ selected coordinates yields $$\frac{1}{n} \ex{}{Z} = \ex{}{\sum_{j \in [d] ~:~ M(X)^j = 1} \overline{X}^j - P^j} \approx k \cdot \ex{}{\overline{X}^j - P^j~\bigg|~\overline{X}^j \geq \frac12 + \Omega(1)} \geq \Omega(\beta k /n),$$ as desired.  %Since $Z$ is this quantity summed over the $n$ samples in the dataset, we obtain $\ex{}{Z} = \Omega(k \beta)$.  
Unfortunately, our actual proof is somewhat more technical and deviates significantly from this intuition, but also gives a more versatile result.

Finally, combining the bounds $\Omega(k \beta) \leq \ex{}{Z} \leq O(n\sqrt{k})$ yields $n \geq \Omega(\sqrt{k} \beta)$ (Theorem \ref{thm:general}).

\subsection{Relationship to Previous Lower Bounds and Attacks}

Our argument is closely related to the work on \emph{tracing attacks}~\cite{H+08, SOJH09, BunUV14, SteinkeU15b, DworkSSUV15, BunSU17, DworkSSU17}.  In a tracing attack, the adversary is given (i) the output $M(X)$ (where $X \leftarrow \cP^n$ consists of $n$ independent samples of individuals' data), (ii) an approximate population mean $p \approx \ex{}{\cP}$, and (iii) the data $Y$ of a ``target'' individual. \tnote{Unmatched parentheses 1) upset my ocd, 2) don't parse in my editor, and (3 were invented by satan himself. Also, using roman to distinguish from $\backslash$eqref.} The target individual is either a random member of the dataset $X$ or an independent random sample from the population $\cP$, and the attacker's goal is to determine which of these two is the case.  Although we don't state our attack in this model, our attack has essentially this format.  Specifically, we consider the quantity $Z=Z_{Y,M(X),p} = \langle M(X), Y - p \rangle$. If $Y \leftarrow \cP$ is a fresh sample from the population, $Z$ is zero in expectation and small with high probability. Whereas, when $Y = X_i$ for a random $i \in [n]$, $Z$ is large in expectation, thus we have some ability to distinguish between these two cases.

Another line of work proves lower bounds in differential privacy via \emph{reconstruction attacks}~\cite{DinurN03, DworkMT07, DworkY08, KasiviswanathanRSU10, KasiviswanathanRS13, MuthukrishnanN12, NikolovTZ13}.  At a high-level, in a reconstruction attack, each sample $X_i$ contains some public information and an independent, random sensitive bit.  The attacker is given $M(X)$ and the public information, and must determine the sensitive bit for $99\%$ of the samples.  These attacks do not give any asymptotic separations between the sample complexities of private and non-private problems, because it is easy to prevent reconstruction attacks without providing meaningful privacy by simply throwing out half of the samples and then running a non-private algorithm on the remaining samples. This subsampling prevents reconstruction and only increases the sample complexity by a factor of two compared to the non-private setting.

%%%%%%%%%
%%%%%%%%%
\jnote{I commented out the paragraph about the adaptive setting for the submission.}\tnote{I worry that a reviewer will say that our lower bound is implied by BUV.}
%\begin{comment}
%%%%%%%%%
%%%%%%%%%
The work of Bun, Ullman, and Vadhan \cite{BunUV14} combines tracing attacks with reconstruction attacks to prove tight lower bounds for large, structured sets of queries (e.g.~all $k$-wise conjunctions). In particular, their work demonstrates that the private multiplicative weights algorithm \cite{HardtR10} is nearly optimal. Since selection is a subroutine of private multiplicative weights, this implies a lower bound for private selection. However, this implicit lower bound for private selection only holds in a complex \emph{adaptive} setting \cite{BunSU17}, where the algorithm must select $k$ items one at a time and the values of the available items change after each selection is made. In contrast, our lower bound is stronger, as it holds for a simple set of items with fixed values.\tnote{Want to mention this, but make it clear our results are much stronger.}
%%%%%%%%%
%%%%%%%%%
%\end{comment}
%%%%%%%%%
%%%%%%%%%

For the special case of pure differential privacy (i.e.~$(\varepsilon,\delta)$-differential privacy with $\delta=0$) lower bounds can be proved using the ``packing'' technique \cite{FeldmanFKN09,BeimelKN10,HardtT10}.
The sample complexity of the top-$k$ selection problem becomes $n=\Theta(k \log d)$ under pure differential privacy. (The upper bound is still attained by repeated use of the exponential mechanism, but the stricter privacy requirement changes the analysis and increases the sample complexity.)  Packing arguments do not provide any non-trivial lower bounds for general differentially private algorithms (i.e.~$(\varepsilon, \delta)$-differential privacy with $\delta > 0$).

\jnote{Not to deny you citations, but I wanted to comment out the CDP lower bound.  It's not a published result and there is no upside to explaining why such a close variant of our problem can be solved by known techniques.}
%For concentrated differential privacy \cite{DworkR16,BunS16} (another special case of differential privacy that is not as restrictive as pure differential privacy), a more involved packing-type lower bound can be used to show that the sample complexity of top-$k$ selection is $n=\Theta(\sqrt{k} \log d)$. However, packing lower bounds do not correspond to realistic privacy attacks.\tnote{Packing should be mentioned, but I don't want to dwell.}

\section{Preliminaries} \label{sec:prelims}

\subsection{Notational Conventions}
We will use the following notational conventions extensively throughout our analysis.  We use $[n] = \{1, 2, \cdots, n\}$ to denote the first $n$ natural numbers.  We use $X \leftarrow \mathcal{D}$ to denote that $X$ is sampled from the probability distribution $\mathcal{D}$.   We also use the shorthand $X_{1 \cdots n} \leftarrow \mathcal{D}$ to denote that $X_1, \cdots, X_n$ are drawn independently from $\mathcal{D}$.  Given a probability $p \in [0,1]$, we use the shorthand $X \leftarrow p$ to denote that $X \leftarrow \mathsf{Bernoulli}(p)$ is a sample from a Bernoulli distribution.   Likewise, $X_{1 \cdots n} \leftarrow p$ denotes that $X_1,\dots,X_n$ are independent samples from $\mathsf{Bernoulli}(p)$. We follow the convention that upper case (non-caligraphic) letters represent random variables and lower case letters represent their realizations. We will treat $X \in (\{0,1\}^d)^n$ and $X \in \{0,1\}^{n \times d}$ equivalently.  For $i \in [n], j \in [d]$, we will subscript $X_i$ to denote the $i^\text{th}$ row, superscript $X^j$ to denote the $j^\text{th}$ column, and $X_i^j$ to denote the entry in the $i^\text{th}$ row and $j^\text{th}$ column, for $i \in [n]$ and $j \in [d]$. We use $\log$ to denote the natural logarithm, i.e.~$\log(z) := \log_e(z)$.

\subsection{Differential Privacy} \label{sec:prelims:dp}
A dataset $x = (x_1,\dots,x_n) \in (\zo^{d})^n$ is an $n \times d$ matrix.  We say that two datasets $x, x'$ are \emph{neighbors} if they differ on at most one row.
\begin{defn}[Differential Privacy \cite{DworkMNS06}]
Fix $n,d \in \N$, $\eps, \delta > 0$.  A (randomized) algorithm $M \from (\zo^d)^n \to \mathcal{R}$ is $(\eps, \delta)$-differentially private if, for every pair of neighboring datasets $x,x'$, and every $R \subseteq \mathcal{R}$,
$$
\pr{}{M(x) \in R} \leq e^{\eps} \pr{}{M(x') \in R} + \delta.
$$
\end{defn}

This definition provides meaningful privacy roughly when $\eps \leq 1$ and $\delta \ll \frac{1}{n}$~\cite{KasiviswanathanS14}.  Since our lower bounds allow for $\eps = 1$ and $\delta$ almost as large as $\frac{1}{n}$, they apply to nearly the entire range of parameters for which differential privacy is meaningful.

\subsection{Beta Distributions} \label{sec:prelims:beta}
Our results make heavy use of the properties of beta distributions.  A \emph{beta distribution}, denoted $\BetaD(\alpha,\beta)$, is a continuous distribution on $[0,1]$ with two parameters $\alpha > 0$ and $\beta > 0$ and probability density at $p$ proportional to $p^{\alpha-1} (1-p)^{\beta-1}$.
More precisely, the cumulative distribution function is described by $$\forall \alpha > 0 ~ \forall \beta > 0 ~ \forall p_* \in [0,1] ~~~~~~~~~~ \pr{P \leftarrow \BetaD(\alpha,\beta)}{P \leq p_*} = \int_0^{p_*} \frac{p^{\alpha-1} (1-p)^{\beta-1}}{\BetaF(\alpha,\beta)} \mathrm{d}p,$$ where $\BetaF(\alpha,\beta) := \int_0^1 p^{\alpha-1} (1-p)^{\beta-1} \mathrm{d} p$ is the beta function.
For all $\alpha,\beta>0$, $$\ex{P \leftarrow \BetaD(\alpha,\beta)}{P} = \frac{\alpha}{\alpha+\beta}~~~\textrm{and}~~~\var{P \leftarrow \BetaD(\alpha,\beta)}{P} = \frac{\alpha\beta}{(\alpha+\beta)^2(\alpha+\beta+1)}.$$ Note that $\BetaD(1,1)$ is simply the uniform distribution on $[0,1]$, and $\BetaD(\beta,\beta)$ becomes more concentrated around its mean of $1/2$ as $\beta$ gets larger.

The key result we need is a form of anti-concentration for beta distributions, which says that if we draw $d$ independent samples from a certain beta distribution with mean $1/2$, then in expectation the $k$ largest samples are at least $3/4$.

\newcommand{\minex}{28} %setting some constant
\begin{prop}\label{prop:betatopk}\tnote{This lemma can definitely have better constants}\tnote{If we want tight bounds for $\alpha$-accuracy, this lemma needs to be generalized to $\beta\approx\alpha \log(d/k)$ and $\ex{}{\max} \geq (1/2+\alpha)k$.} \jnote{Good thing to keep in mind.  Maybe we should do that for some future version?  I am kinda rushing to the FOCS deadline and don't want to make too many fundamental changes that require global edits.}
Fix $\beta > 0$ and $d,k \in \mathbb{N}$. Let $P^1, \cdots, P^d$  be independent samples from $\BetaD(\beta,\beta)$. If $k \geq 1$ and $1 \leq \beta \leq 1+\frac{1}{2} \log\left(\frac{d}{8 \max\{2k,\minex\}}\right)$, then $$\ex{P^{1 \cdots d}}{\max_{s \subset [d] ~:~ |s|=k} \sum_{j \in s} P^j} \geq \frac{3}{4} k.$$ 
\end{prop}

The above proposition follows from an anti-concentration lemma for the beta distribution.
\begin{lem}\label{lem:BetaL}
For all $\beta \geq 1$ and all $p_* \in [0,1/2]$, $$\pr{P \leftarrow \BetaD(\beta,\beta)}{P > 1 - p_*} =\pr{P \leftarrow \BetaD(\beta,\beta)}{P < p_*} \geq \left( 4 p_* (1-p_*) \right)^{\beta-1} \frac{p_*}{\beta} \geq p_* \cdot e^{(\log(4p_*(1-p_*))-1)(\beta-1)}.$$
\end{lem}
\begin{proof}[Proof of Lemma~\ref{lem:BetaL}]
The equality in the statement follows from the fact that $\BetaD(\beta,\beta)$ is symmetric around $1/2$. Now we prove the inequalities.
We use two bounds: Firstly, $p(1-p) \leq 1/4$ for all $p \in [0,1]$. Secondly, $p(1-p) \geq p(1-p_*)$ for all $p \in [0,p_*]$. Thus
\begin{align*}
\pr{P \leftarrow \BetaD(\beta,\beta)}{P < p_*} = \frac{\int_0^{p_*} (p(1-p))^{\beta-1} \mathrm{d}p}{\int_0^1 (p(1-p))^{\beta-1} \mathrm{d} p} 
\geq{}& \frac{\int_0^{p_*} (p(1-p_*))^{\beta-1} \mathrm{d}p}{\int_0^1 (1/4)^{\beta-1} \mathrm{d} p}\\
={}&  \frac{(1-p_*)^{\beta-1} \int_0^{p_*} p^{\beta-1} \mathrm{d}p}{(1/4)^{\beta-1} }\\
={}& \left( 4(1-p_*) \right)^{\beta-1} \frac{p_*^\beta}{\beta}\\
={}& \left( 4 p_* (1-p_*) \right)^{\beta-1} \frac{p_*}{\beta}=: f_{p_*}(\beta).
\end{align*}
This proves the first inequality. Now we prove the second inequality by applying calculus to the function $f_{p_*}$ we have just defined.

We have $f_{p_*}(1) = p_*$ and $$f_{p_*}'(\beta) = \left( 4 p_* (1-p_*) \right)^{\beta-1}\frac{p_*}{\beta} \left( \log(4p_*(1-p_*)) - \frac{1}{\beta} \right) \geq f_{p_*}(\beta) \left( \log(4p_*(1-p_*))-1 \right).$$
This differential inequation implies $f_{p_*}(\beta) \geq p_* e^{(\log(4p_*(-p_*))-1)(\beta-1)}$, as required.
\end{proof}

\begin{comment}%%%%%%%%%%%%%%%%%%%%%%%%%
We also have a nearly-matching concentration result for beta distributions:
\begin{lem}\label{lem:BetaU}
For all $\beta \geq 1$ and all $p_* \in [0,1/2]$, $$\pr{P \leftarrow \BetaD(\beta,\beta)}{P > 1 - p_*} =\pr{P \leftarrow \BetaD(\beta,\beta)}{P < p_*} \leq p_* 4^\beta.$$
\end{lem}
\begin{proof}[Proof of Lemma \ref{lem:BetaU}]
For all $\varepsilon \in (0,1)$, we have
\begin{align*}
\pr{P \leftarrow \BetaD(\beta,\beta)}{P < p_*} &= \frac{\int_0^{p_*} (p(1-p))^{\beta-1} \mathrm{d}p}{\int_0^1 (p(1-p))^{\beta-1} \mathrm{d}p}\\
&\leq \frac{\int_0^{p_*} (p \cdot 1)^{\beta-1} \mathrm{d}p}{\int_{\frac12 - \frac12 \varepsilon}^{\frac12 + \frac12 \varepsilon} (p(1-p))^{\beta-1} \mathrm{d}p}\\
&\leq \frac{\int_0^{p_*} p^{\beta-1} \mathrm{d}p}{\int_{\frac12 - \frac12 \varepsilon}^{\frac12 + \frac12 \varepsilon} ((1/2-\varepsilon/2)(1/2+\varepsilon/2))^{\beta-1} \mathrm{d}p}\\
&= \frac{p_*^\beta / \beta}{\varepsilon \cdot ((1-\varepsilon^2)/4)^{\beta-1}}\\
&= \frac{p_*^\beta}{\beta} 4^{\beta-1} \frac{(1-\varepsilon^2)^{1-\beta}}{\varepsilon}\\
\left[1-\varepsilon^2 \geq e^{-2\varepsilon^2}\right]~~&\leq \frac{p_*^\beta}{\beta} 4^{\beta-1} \frac{e^{2\varepsilon^2(\beta-1)}}{\varepsilon}.
\end{align*}
We set $\varepsilon=1/\sqrt{4(\beta-1)}$ to obtain $$\pr{P \leftarrow \BetaD(\beta,\beta)}{P < p_*} \leq \frac{p_*^\beta}{\beta} 4^{\beta-1} \cdot 2e^{1/2}(\beta-1) \leq p_* 4^{\beta}.$$
\end{proof}
\end{comment}%%%%%%%%%%%%%%%%%%%%%%%%%%

\begin{proof}[Proof of Proposition \ref{prop:betatopk}]
By Lemma \ref{lem:BetaL}, for each $j \in [d]$, 
$$
\pr{P^j \leftarrow \BetaD(\beta,\beta)}{P^j > \frac{7}{8}} \geq \frac{1}{8} e^{(\log(4 \cdot (1/8) \cdot (7/8))-1) (\beta-1)} \geq \frac{1}{8} e^{-2(\beta-1)},
$$ 
Since $1 \leq \beta \leq 1+ \frac{1}{2} \log\left(\frac{d}{8\max\{2k,\minex\}}\right) \geq 1$, we have, for all $j \in [d]$, $$\pr{P^j \leftarrow \BetaD(\beta,\beta)}{P^j > \frac{7}{8}} \geq \frac{\max\{2k,\minex\}}{d}.$$
Let $A_j = \mathbb{I}[P^j > 7/8]$ be the indicator of the above event for each $j \in [d]$. Let $a:=\ex{}{A_1}=\ex{}{A_2}=\cdots=\ex{}{A_d}$. Then $\ex{}{\sum_{j \in [d]} A_j} = d a \geq \max\{2k,\minex\}$. Since $A_j \in \{0,1\}$, $\ex{}{A_j^2}=\ex{}{A_j}=a$ and $\var{}{A_j} =a(1-a)$ for all $j \in [d]$ and, hence, $\var{}{\sum_{j \in [d]} A_j} = d a (1-a)$. By Chebyshev's inequality, 
\begin{align*}
&\pr{}{\exists s \subset [d] ~~~ |s|=k ~\wedge~ \forall j \in s ~~P^j>7/8} \\
={} &\pr{}{\sum_{j \in [d]} A_j \geq k} \\
\geq{} &\pr{}{\sum_{j \in [d]} A_j \geq \frac12 \max\{2k,\minex\}} \\
\geq{} &\pr{}{\sum_{j \in [d]} A_j \geq \frac12 \ex{}{\sum_{j \in [d]} A_j}}\\
\geq{} &1- \pr{}{\left| \sum_{j \in [d]} A_j - \ex{}{\sum_{j \in [d]} A_j} \right| \geq \frac12 \ex{}{\sum_{j \in [d]} A_j}}\\
\geq{} &1 - \frac{\var{}{\sum_{j \in [d]} A_j}}{\left(\frac12 \ex{}{\sum_{j \in [d]} A_j}\right)^2}\\
={} &1 - 4 \frac{d a (1-a)}{(da)^2} \\
\geq{} &1 - \frac{4}{da}
\geq{} 1 - \frac{4}{\max\{2k,\minex\}}
\geq{} \frac 6 7.
\end{align*}
Finally, 
\begin{align*}
\ex{}{\max_{s \subset [d] ~:~ |s|=k} \sum_{j \in s} P^j} &\geq
 \frac{7}{8} k \pr{}{\exists s \subset [d] ~~~ (|s|=k)~\wedge~(\forall j \in s ~~P^j>7/8)} \\
 &\geq  \frac{7}{8} k \frac{6}{7} = \frac{3}{4} k.
\end{align*}
This completes the proof.
\end{proof}

%%%%%%%%%%%%%%%%%%%%%%%%
%%%%%%%%%%%%%%%%%%%%%%%%
\jnote{empirical vs.~distributional is in the comments.}

\section{Proof of the Main Lower Bound (Theorem \ref{thm:general})}
The goal of this section is to prove the following theorem from the introduction.
\begin{thm}[Theorem \ref{thm:general} restated]\label{thm:general-body}
Let $\beta,\gamma,k>0$ and $n,d \in \mathbb{N}$ be a fixed set of parameters.  Let $P^1, \cdots, P^d$ be independent draws from $\BetaD(\beta,\beta)$ and let $X \in (\{0,1\}^d)^n$ be a random dataset such that every $X_i^j$ is independent (conditioned on $P$) and $\E[X_i^j]=P^j$ for every $i \in [n]$ and $j \in [d]$.

Let $M : (\{0,1\}^d)^n \to [-1,1]^d$ be a $(1,\beta\gamma k / n d)$-differentially private algorithm and assume that $M$ satisfies the condition $\ex{P,X,M}{\|M(X)\|_2^2} = k$ and the accuracy condition
\begin{equation*}
\ex{P,X,M}{\sum_{j \in [d]} M(X)^j \cdot \left(P^j - \frac12\right)} \geq \gamma k.
\end{equation*}
Then $n \geq \gamma \beta \sqrt{k}.$
\end{thm}

For the remainder of this section, we will fix the following parameters and variables.  Fix $\beta,\gamma,\varepsilon,\delta,\Delta>0$ and $n,d \in \mathbb{N}$.  Let $M : (\{0,1\}^d)^n \to \mathbb{R}^d$ satisfy $(\varepsilon,\delta)$-differential privacy. Let $P^1,\cdots,P^d \leftarrow \BetaD(\beta,\beta)$ be independent samples. Define a random variable $X \in (\{0,1\}^d)^n$ to have independent entries (conditioned on $P$) where $\mathbb{E}[X_i^j]=P^j$ for all $i \in [n]$ and $j \in [d]$. 

The crux of the proof is to analyze the expected value of $\sum_{i \in [n], j \in [d]} M(X)^j (X_i^j - P^j).$  To this end, for every $i \in [n]$ and $j \in [d]$, we define the random variables $$Z_i^j = M(X)^j \cdot (X_i^j - P^j), ~~~~~ Z_i = \sum_{j \in [d]} Z_i^j, ~~~~~ Z^j = \sum_{i \in [n]} Z_i^j, ~~~~~ Z = \sum_{i \in [n] \atop j \in [d]} Z_i^j.$$  At a high level, we will show that when the size of the dataset $n$ is too small, we obtain contradictory upper and lower bounds on $\ex{P,X,M}{Z}$.

\subsection{Upper Bound via Privacy}
First we prove that $(\eps, \delta)$-differential privacy of $M$ implies an upper bound on $\ex{P,X,M}{Z}$.

\tnote{We can tighten the following analysis to get our lovely $\log(1/\delta)$ dependence. The step for improvement is where the $\min\{1,\cdot\}$ gets dropped.}

\begin{lem}\label{lem:upper}
Suppose that $\|M(X)\|_1\leq 2\Delta$ with probability 1. Then
$$\ex{P,X,M}{Z} \leq n \cdot \left(e^\varepsilon \frac12 \sqrt{\ex{P,X,M}{\|M(X)\|_2^2}} + \Delta\delta\right).$$
\end{lem}
\begin{proof}
Fix $i \in [n]$ and also fix the vector $P \in [0,1]^d$.
Since $\sum_{j \in [d]} |M(X)^j| \leq 2\Delta$, we have 
$$
Z_i = \sum_{j \in [d]} M(X)^j \left(P^j - \frac12\right) \leq \Delta.$$
Let $X_{\sim i} \in (\{0,1\}^d)^n$ denote $X$ with $X_i$ replaced with an independent draw from $P$. In particular, the marginal distribution of $X_{\sim i}$ is the same as $X$. However, conditioned on $P$, $X_i$ is independent from $X_{\sim i}$. By the differential privacy assumption $M(X)$ and $M(X_{\sim i})$ are indistinguishable. We can use this fact to bound the expectation of $Z_i$ in the following calculation.
\begin{align}
\ex{X,M}{Z_i} &\leq \ex{X,M}{\max\{0,Z_i\}} \notag \\
&= \int_0^\Delta \pr{X,M}{Z_i \geq z} \mathrm{d}z \notag \\
&= \int_0^\Delta  \pr{X,M}{\sum_{j \in [d]} M(X)^j \left(X_i^j - P^j\right) \geq z} \mathrm{d}z \label{eq:longcalc1}
\end{align}
Now, defining the event $T(z) :=  \left\{y \in \mathbb{R}^d : \sum_{j \in [d]} y^j \left(X_i^j - P^j\right) \geq z \right\}$, we can apply $(\eps, \delta)$-differential privacy to \eqref{eq:longcalc1} to obtain
\begin{align}
\eqref{eq:longcalc1}
&= \int_0^\Delta  \pr{X,M}{M(X) \in T(z)} \mathrm{d}z \notag \\
&\leq \int_0^\Delta \min\left\{ 1,~e^\varepsilon \pr{X,X_{\sim i},M}{M(X_{\sim i}) \in T(z)} +\delta \right\}\mathrm{d}z \notag \\
&= \int_0^\Delta \min\left\{ 1,~e^\varepsilon \pr{X,X_{\sim i},M}{\sum_{j \in [d]} M(X_{\sim i})^j \left(X_i^j - P^j\right) \geq z } +\delta \right\}\mathrm{d}z \label{eq:longcalc2}
\end{align}
Observe that in \eqref{eq:longcalc2}, $M(X_{\sim i})$ is independent of $X_i$, which allows us to bound \eqref{eq:longcalc2} as follows
\begin{align*}
\eqref{eq:longcalc2} 
&\leq \int_0^\Delta e^\varepsilon \pr{X,X_{\sim i},M}{\sum_{j \in [d]} M(X_{\sim i})^j \left(X_i^j - P^j\right) \geq z } +\delta \mathrm{d}z \\
&= e^\varepsilon \ex{X,X_{\sim i},M}{ \max\left\{0,\sum_{j \in [d]} M(X_{\sim i})^j \left(X_i^j - P^j\right) \right\}} + \Delta \delta\\
&\leq e^\varepsilon \sqrt{\ex{X,X_{\sim i},M}{ \left(\sum_{j \in [d]} M(X_{\sim i})^j \left(X_i^j - P^j\right) \right)^2}} + \Delta \delta\\
&= e^\varepsilon \sqrt{\ex{X_{\sim i},M}{ \sum_{j \in [d]} (M(X_{\sim i})^j)^2 \ex{X_i}{\left(X_i^j - P^j\right)^2}}} + \Delta \delta\\
&\leq e^\varepsilon \sqrt{\ex{X_{\sim i},M}{ \sum_{j \in [d]} (M(X_{\sim i})^j)^2 \frac{1}{4}}} + \Delta \delta
= e^\varepsilon \frac12 \sqrt{\ex{X,M}{\|M(X)\|_2^2} } + \Delta \delta.
\end{align*}

Finally, we sum over $i$ and average over $P$ to obtain
$$\ex{P,X,M}{Z} \leq n \cdot \ex{P}{e^\varepsilon \frac12 \sqrt{\ex{X,M}{\|M(X)\|_2^2} } + \Delta \delta} \leq n \cdot \left( e^\varepsilon \frac12 \sqrt{\ex{P,X,M}{\|M(X)\|_2^2} } + \Delta \delta\right),$$
where the final inequality follows from Jensen's inequality and the concavity of the function $x \mapsto \sqrt{x}$.
\end{proof}

\subsection{Lower Bound via Accuracy}
The more involved part of the proof is to use the accuracy assumption 
\begin{equation*}
\ex{P,X,M}{\sum_{j \in [d]} M(X)^j \cdot \left(P^j - \frac12\right)} \geq \gamma \cdot \ex{P,X,M}{\|M(X)\|_2^2}.
\end{equation*}
to prove a lower bound on $\ex{P,X,M}{Z}$.  In order to do so we need to develop a technical tool that we call a ``fingerprinting lemma,'' which is a refinement of similar lemmas from prior work~\cite{SteinkeU15b, DworkSSUV15, BunSU17} that more carefully exploits the properties of the distribution $P$.

\subsubsection{Fingerprinting Lemma}

To keep our notation compact, throughout this section we will use the shorthand $X_{1 \cdots n} \leftarrow p$ to denote that $X_1, X_2, \cdots, X_n \in \{0,1\}$ are independent random variables each with mean $p$.

%First we rescale the fingerprinting derivative lemma \cite{SteinkeU15a,DworkSSUV15} to $\{0,1\}$ (rather than $\{-1,1\}$), as this is where the beta distribution is more familiar:

\begin{lem} [Rescaling of~\cite{SteinkeU15a, DworkSSUV15}]\label{lem:FPCd}
Let $f : \{0,1\}^n \to \mathbb{R}$. Define $g : [0,1] \to \mathbb{R}$ by 
$$
g(p) = \ex{X_{1 \cdots n} \leftarrow p}{f(X)}.
$$
Then
$$
\ex{X_{1 \cdots n} \leftarrow p}{f(X) \sum_{i \in [n]} (X_i-p)} = p (1-p) g'(p)
$$
for all $p \in [0,1]$.
\end{lem}
\begin{proof}[Proof of Lemma \ref{lem:FPCd}]
Firstly, $\pr{X \leftarrow p}{X=1} = p$ and $\pr{X \leftarrow p}{X=0} = 1-p$. Thus \begin{equation}p (1-p)\frac{\mathrm{d}}{\mathrm{d} p} \pr{X \leftarrow p}{X=1}  = p (1-p)\frac{\mathrm{d}}{\mathrm{d} p} p = p (1-p) = (1-p) \pr{X \leftarrow p}{X=1}\label{eqn:1}\end{equation} and \begin{equation}p(1-p)\frac{\mathrm{d}}{\mathrm{d} p} \pr{X \leftarrow p}{X=0} =  p(1-p) \frac{\mathrm{d}}{\mathrm{d} p}  (1-p)= p(1-p)(-1) = (0-p) \pr{X \leftarrow p}{X=0}.\label{eqn:0}\end{equation}
Hence
\begin{align*}
p(1-p) g'(p) %=& \frac{\mathrm{d}}{\mathrm{d} p} \ex{X_{1 \cdots n} \leftarrow p}{f(X)} \\
={}& p (1-p) \frac{\mathrm{d}}{\mathrm{d} p} \sum_{x \in \{0,1\}^n} \pr{X_{1 \cdots n} \leftarrow p}{X=x} f(x) \\
={}& \sum_{x \in \{0,1\}^n} f(x) p (1-p)\frac{\mathrm{d}}{\mathrm{d} p} \prod_{i \in [n]} \pr{X \leftarrow p}{X=x_i}\\
={}& \sum_{x \in \{0,1\}^n} f(x) \sum_{i \in [n]} \left( \prod_{j \in [n] \setminus \{i\}} \pr{X \leftarrow p}{X=x_j} \right) \left( p (1-p)\frac{\mathrm{d}}{\mathrm{d} p} \pr{X \leftarrow p}{X=x_i} \right)\\
\left(\text{by \eqref{eqn:1} and \eqref{eqn:0}}\right)~~~={}& \sum_{x \in \{0,1\}^n} f(x) \sum_{i \in [n]} \left( \prod_{j \in [n] \setminus \{i\}} \pr{X \leftarrow p}{X=x_j} \right) \left( (x_i-p) \pr{X \leftarrow p}{X=x_i} \right)\\
={}& \sum_{x \in \{0,1\}^n} f(x) \sum_{i \in [n]} (x_i -p) \left( \prod_{j \in [n]}  \pr{X \leftarrow p}{X=x_j} \right) \\
={}& \ex{X_{1 \cdots n} \leftarrow p}{f(X) \sum_{i \in [n]} (X_i-p)}.
\end{align*}
\end{proof}

\begin{lem}\label{lem:FPCb}
Let $f : \{0,1\}^n \to \mathbb{R}$ and let $\alpha,\beta > 0$. Define $g : [0,1] \to \mathbb{R}$ by $$g(p) = \ex{X_{1 \cdots n} \leftarrow p}{f(X)}.$$ Then $$\ex{P \leftarrow \BetaD(\alpha,\beta) \atop X_{1 \cdots n} \leftarrow P}{f(X) \sum_{i \in [n]} (X_i-P)} =  (\alpha+\beta)\ex{P \leftarrow \BetaD(\alpha,\beta)}{g(P)\left(P-\frac{\alpha}{\alpha+\beta}\right)}.$$\end{lem}
This is the form of the lemma we use. Note that $\ex{}{P}=\alpha/(\alpha+\beta)$.
\begin{proof}[Proof of Lemma \ref{lem:FPCb}]
The proof is a calculation using integration by parts.  Using Lemma \ref{lem:FPCd} and the fundamental theorem of calculus, we have
\begin{align*}
\ex{P \leftarrow \BetaD(\alpha,\beta) \atop X_{1 \cdots n} \leftarrow P}{f(X) \sum_{i \in [n]} (X_i-P)} =& \ex{P \leftarrow \BetaD(\alpha,\beta)}{P(1-P) g'(P)}\\
={}& \int_0^1 p(1-p)g'(p) \cdot \frac{p^{\alpha-1}(1-p)^{\beta-1}}{\BetaF(\alpha,\beta)} \mathrm{d}p\\
={}& \frac{1}{{\BetaF(\alpha,\beta)}}\int_0^1 g'(p) \cdot {p^{\alpha}(1-p)^{\beta}} \mathrm{d}p\\
={}& \frac{1}{{\BetaF(\alpha,\beta)}}\int_0^1 \left( \frac{\mathrm{d}}{\mathrm{d}p} \left( g(p) \cdot {p^{\alpha}(1-p)^{\beta}} \right) - g(p) \cdot \frac{\mathrm{d}}{\mathrm{d}p} \left( p^\alpha (1-p)^\beta \right) \right)\mathrm{d}p\\
={}& \frac{1}{{\BetaF(\alpha,\beta)}}\left(  g(1) \cdot 1^{\alpha}(1-1)^{\beta} - g(0) \cdot 0^{\alpha}(1-0)^{\beta} \right) \\
&- \frac{1}{{\BetaF(\alpha,\beta)}}\int_0^1  g(p) \cdot \frac{\mathrm{d}}{\mathrm{d}p} \left( p^\alpha (1-p)^\beta \right)\mathrm{d}p\\
={}&0- \frac{1}{{\BetaF(\alpha,\beta)}}\int_0^1  g(p) \cdot (\alpha-(\alpha+\beta)p) p^{\alpha-1} (1-p)^{\beta-1})\mathrm{d}p\\
={}&  \ex{P \leftarrow \BetaD(\alpha,\beta)}{g(P) ((\alpha+\beta)P - \alpha)}.
\end{align*}
This completes the proof.
\end{proof}

%%%%%%%%%
%%%%%%%%%
\jnote{There is a potentially useful cor commented out here.}
\begin{comment}
%%%%%%%%%
%%%%%%%%%
\begin{cor}
Let $f : \{0,1\}^n \to \mathbb{R}$ and let $\alpha,\beta > 0$.  Then $$\left( 1 + \frac{\alpha+\beta}{n} \right)\ex{P \leftarrow \BetaD(\alpha,\beta) \atop X_{1 \cdots n} \leftarrow P}{f(X) \sum_{i \in [n]} (X_i-P)} =  (\alpha+\beta)\ex{P \leftarrow \BetaD(\alpha,\beta) \atop X_{1 \cdots n} \leftarrow P}{f(X)\left(
\frac{1}{n} \sum_{i \in [n]} X_i -\frac{\alpha}{\alpha+\beta}\right)}.$$
\end{cor}
\begin{proof}
Define $g : [0,1] \to \mathbb{R}$ by $$g(p) = \ex{X_{1 \cdots n} \leftarrow p}{f(X)}.$$
By Lemma \ref{lem:FPCb}, 
\begin{align*}
\ex{P \leftarrow \BetaD(\alpha,\beta) \atop X_{1 \cdots n} \leftarrow P}{f(X) \sum_{i \in [n]} (X_i-P)} 
=&  (\alpha+\beta)\ex{P \leftarrow \BetaD(\alpha,\beta)}{g(P)\left(P-\frac{\alpha}{\alpha+\beta}\right)}\\
=&  (\alpha+\beta)\ex{P \leftarrow \BetaD(\alpha,\beta) \atop X_{1 \cdots n} \leftarrow P}{f(X)\left(P - \frac{1}{n} \sum_{i \in [n]} X_i + \frac{1}{n} \sum_{i \in [n]} X_i -\frac{\alpha}{\alpha+\beta}\right)}\\
=&  (\alpha+\beta)\ex{P \leftarrow \BetaD(\alpha,\beta) \atop X_{1 \cdots n} \leftarrow P}{f(X)\left(\frac{1}{n} \sum_{i \in [n]} X_i -\frac{\alpha}{\alpha+\beta}\right)}\\
& -  \frac{1}{n} (\alpha+\beta)\ex{P \leftarrow \BetaD(\alpha,\beta) \atop X_{1 \cdots n} \leftarrow P}{f(X) \sum_{i \in [n]} \left( X_i -P \right)}.
\end{align*}
Rearranging yields the result.
\end{proof}
%%%%%%%%%
%%%%%%%%%
\end{comment}
%%%%%%%%%
%%%%%%%%%

\subsubsection{Using the Fingerprinting Lemma}

Now we can use Lemma \ref{lem:FPCb} to prove a lower bound 

\begin{lem}\label{lem:lower}
$$\ex{P,X,M}{Z} \geq 2\beta \ex{P,X,M}{\sum_{j \in [d]} M(X)^j \left(P^j - \frac12\right)}$$
\end{lem}
\begin{proof}[Proof of Lemma \ref{lem:lower}]
Fix a column $j \in [d]$.  Define $f : \zo^n \to [0,1]$ and $g : [0,1] \to [0,1]$ to be
$$
f(x^j) := \ex{P^{-j},X^{-j}}{M(x^{-j} \| X^{j})^j}
$$
and define $g$ to be
$$
g(p^j) := \ex{P^{-j},X^j_{1\dots n} \sim P^j}{f(X^j)}.
$$
That is $f(x)$ is the expectation of $M(X)^j$ conditioned on $X^j=x$, where the expectation is over the randomness of $M$ and the randomness of $P^{j'}$ and $X^{j'}$ for $j' \ne j$. Similarly $g(p)$ is the expectation of $M(X)^j$ conditioned on $P^j=p$, where the expectation is over $M$ and $P^{j'}$ and $X^{j'}$ for $j' \ne j$ and also over $X^j$.
%where $P^{-j},X^{-j}$ denote the variables $P,X$ with the $j$-th column removed and $(X^{-j} \| X^{j})$ denotes that $X^{j}$ is inserted into $X^{-j}$ as the $j$-th column.
Now we can calculate
\begin{align*}
\ex{P,X,M}{Z^j} &= \ex{P,X,M}{M(X)^j \sum_{i \in [n]} (X_i^j - P^j)}\\
&= \ex{P^j \leftarrow \BetaD(\beta,\beta) \atop X^j_{1 \cdots n} \leftarrow P^j}{f(X^j) \sum_{i \in [n]} (X^j_i - P^j)}\\
&= 2\beta \ex{P^j \leftarrow \BetaD(\beta,\beta)}{g(P^j) \left( P^j - \frac12 \right)} \tag{Lemma \ref{lem:FPCb}} \\
&= 2\beta \ex{P^j \leftarrow \BetaD(\beta,\beta) \atop X^j_{1 \cdots n} \leftarrow P^j}{f(X^j) \left( P^j - \frac12 \right)}\\
&= 2\beta \ex{P,X,M}{M(X)^j \left( P^j - \frac12 \right)},
\end{align*}
The result now follows by summation over $j \in [d]$.
\end{proof}

%%%%%%%%%
%%%%%%%%%
\begin{comment}
%%%%%%%%%
%%%%%%%%%
\begin{proof}
Fix $j \in [d]$. By Lemma \ref{lem:FPCb}, 
\begin{align*}
\ex{P,X,M}{Z^j} &= \ex{P,X,M}{M(X)^j \sum_{i \in [n]} (X_i^j - P^j)}\\
&= \ex{\hat P \leftarrow \BetaD(\beta,\beta) \atop \hat X_{1 \cdots n} \leftarrow \hat P}{f(\hat X) \sum_{i \in [n]} (\hat X_i - \hat P)}\\
&= 2\beta \ex{\hat P \leftarrow \BetaD(\beta,\beta)}{g(\hat P) \left( \hat P - \frac12 \right)}\\
&= 2\beta \ex{\hat P \leftarrow \BetaD(\beta,\beta) \atop \hat X_{1 \cdots n} \leftarrow \hat P}{f(\hat X) \left( \hat P - \frac12 \right)}\\
&= 2\beta \ex{P,X,M}{M(X)^j \left( P^j - \frac12 \right)},
\end{align*}
where $f(x) = \ex{}{M(X)^j \mid X^j = x}$ and $g(p) = \ex{\hat X_{1 \cdots n} \leftarrow p}{f(\hat X)}$. The result now follows by summation over $j \in [d]$.
\end{proof}
%%%%%%%%%
%%%%%%%%%
\end{comment}
%%%%%%%%%
%%%%%%%%%

\subsection{Putting it Together}

We can now combine the upper bound (Lemma \ref{lem:upper}) and the lower bound (Lemma \ref{lem:lower}) that we've proven on the expectation of $Z$ to complete the proof of Theorem \ref{thm:general}.

%\begin{prop}\label{prop:lb}
%Suppose that $\|M(X)\|_1\leq 2\Delta$ with probability 1 and $$\ex{P,X,M}{\sum_{j \in [d]} M(X)^j \left(P^j - \frac12\right)} \geq \gamma \cdot \ex{P,X,M}{\|M(X)\|_2^2}.$$
%If $\delta \leq \beta\gamma \ex{P,X,M}{\|M(X)\|_2^2}/n\Delta$, then $$n \geq \frac{2\beta\gamma}{e^\varepsilon} \sqrt{\ex{P,X,M}{\|M(X)\|_2^2}} .$$
%\end{prop}
\begin{proof} [Proof of Theorem \ref{thm:general}]
By our accuracy assumption and Lemmas \ref{lem:lower} and \ref{lem:upper}, 
\begin{align*}
&2\beta\gamma \cdot \ex{P,X,M}{\|M(X)\|_2^2} \\
\leq{} &2\beta \ex{P,X,M}{\sum_{j \in [d]} M(X)^j \left(P^j - \frac12\right)} \\
\leq{} &\ex{P,X,M}{Z} \\
\leq{} &n \cdot e^\varepsilon \frac12 \sqrt{\ex{P,X,M}{\|M(X)\|_2^2}} + n\Delta\delta.
\end{align*}
This implies $$n \geq \frac{4\beta\gamma \ex{P,X,M}{\|M(X)\|_2^2} - 2n\Delta\delta}{e^\varepsilon \sqrt{\ex{P,X,M}{\|M(X)\|_2^2}}} \geq \frac{3\beta\gamma}{e} \sqrt{\ex{P,X,M}{\|M(X)\|_2^2}},$$
where the final inequality follows from $\ex{P,X,M}{\|M(X)\|_2^2} =k$, $\|M(x)\|_1 \leq d=2\Delta$, $\varepsilon=1$, and $\delta =\beta\gamma k/nd$. 
\end{proof}

\section{Using our Lower Bound}

In this section we show how to apply the lower bound of Theorem \ref{thm:general} to natural problems, and thereby prove Theorems \ref{thm:intro1} and \ref{thm:intro2} from the introduction.  We can also use it to derive known lower bounds for releasing the dataset mean~\cite{BunUV14, SteinkeU15b, DworkSSUV15}, which we detail in Appendix \ref{sec:meanlb}.

\jnote{I moved the application to releasing the dataset mean to an appendix for the submission.}

\subsection{Application to Approximate Top-$k$ Selection}

We first state an upper bound for the top-$k$ selection problem:

\begin{thm} \label{thm:topkub}
Fix $d, k \in \N$ and $\alpha,\eps, \delta >0$.  For every $n \geq \frac{1}{\alpha\eps}\sqrt{8 k \log(\frac{e^\eps}{\delta})} \log(d)$, there is an $(\eps, \delta)$-differentially private algorithm $M \from (\zo^{d})^{n} \to \zo^{d}$ such that for every $x \in (\zo^{d})^{n}$, such that 
$$\ex{M}{ \sum_{j \in [d]} M(x)^j \overline{x}^j } \geq \max_{s \subset [d] : |s|=k} \sum_{j \in s} \overline{x}^j - \alpha k.$$
\end{thm}
This theorem follows immediately by using the exponential mechanism~\cite{McSherryT07} (see \cite[Theorem 3.10]{DworkR14} and \cite[Lemma 7.1]{BassilyNSSSU16} for the analysis) to repeatedly ``peel off'' the column of $x$ with the approximately largest mean $\overline{x}^{j}$, and applying the composition theorem for differential privacy~\cite{DworkMNS06, DworkRV10,BunS16}.

Alternatively, Theorem \ref{thm:topkub} can provide $\rho$-concentrated differential privacy \cite{BunS16} (instead of $(\varepsilon,\delta)$-differential privacy) for $n\geq\frac{\log d}{\alpha} \sqrt{\frac{2k}{\rho}}$ (with the same accuracy guarantee).\tnote{I'm going to push this damnit.}

Using Theorem \ref{thm:general} we can obtain a nearly matching lower bound that is tight up to a factor of $O(\sqrt{\log(1/\delta)})$ in most parameter regimes.\footnote{The lower bound can be made to include a $\log(1/\delta)$ factor using a group privacy reduction \cite{SteinkeU15b}. For the sake of clarity, we do not delve into this issue.}\tnote{I kind of want to hide the missing $\log(1/\delta)$ factor.}  The lower bound actually holds even for algorithms $M$ that provide just average case accuracy guarantees.
\begin{cor} \label{cor:topklb}
Fix $n,d,k \in \mathbb{N}$ with $d \geq \max\{16k,224\}$.
Set $\beta = 1+ \frac{1}{2} \log\left(\frac{d}{8\max\{2k,28\}}\right)$. Let $P^1, \cdots, P^j$ be independent draws from $\BetaD(\beta,\beta)$ and let $X \in (\{0,1\}^d)^n$ be such that each $X_i^j$ is independent (conditioned on $P$) and $\ex{}{X_i^j}=P^j$ for all $i \in [n]$ and $j \in [d]$. 
Let $M : (\{0,1\}^d)^n \to \{0,1\}^d$ be $(1,1/nd)$-differentially private. Suppose $\|M(x)\|_1 = \|M(x)\|_2^2 = k$ for all $x$ with probability 1.  
Suppose $$\ex{P,X,M}{ \sum_{j \in [d]} M(X)^j P^j } \geq \ex{P}{\max_{s \subset [d] : |s|=k} \sum_{j \in s} P^j} - \frac{k}{10}.$$
Then $n \geq \frac{1}{14} \sqrt{k} \log(\frac{d}{16k+208})$.
\end{cor}
Although Corollary \ref{cor:topklb} is stated for accuracy guarantees that hold with respect to the population mean $P$, since $\ex{X}{\|\overline{X}-P\|_\infty} \leq \sqrt{\frac{\log(2d)}{2n}}$, we can replace the accuracy condition with $$\ex{P,X,M}{ \sum_{j \in [d]} M(X)^j \overline{X}^j } \geq \ex{P,X}{\max_{s \subset [d] : |s|=k} \sum_{j \in s} \overline{X}^j} - k \left(\frac{1}{10} - \sqrt{\frac{2\log(2d)}{n}}\right)$$ to get a theorem that is more directly comparable to Theorem  \ref{thm:topkub}.
\begin{proof}[Proof of Corollary \ref{cor:topklb}]
By Proposition \ref{prop:betatopk} and our choice of $\beta$, $$\ex{P^{1 \cdots d}}{\max_{s \subset [d] ~:~ |s|=k} \sum_{j \in s} P^j} \geq \frac{3}{4} k.$$ 
Thus $$\ex{P,X,M}{ \sum_{j \in [d]} M(X)^j \left( P^j -\frac12 \right)} \geq \frac{1}{4} k - \frac{k}{10} = \frac{3}{20} \ex{P,X,M}{\|M(X)\|_2^2}.$$
Thus, by Theorem \ref{thm:general}, $$n \geq \beta \frac{3}{20} \sqrt{k} = \frac{3}{20} \sqrt{k} \left( 1 + \frac12 \log\left(\frac{d}{8\max\{2k,28\}}\right)\right).$$  This completes the proof.
\end{proof}

\subsection{Application to Multiple Hypothesis Testing}
\begin{cor} \label{cor:mhtlb}
Fix $n,d,k \in \mathbb{N}$ with $d \geq 16k \geq 32$.
Set $\beta = 1+ \frac{1}{2} \log\left(\frac{d}{16k}\right)$. Let $P^1, \cdots, P^j$ be independent draws from $\BetaD(\beta,\beta)$ and let $X \in (\{0,1\}^d)^n$ be such that each $X_i^j$ is independent (conditioned on $P$) and $\ex{}{X_i^j}=P^j$ for all $i \in [n]$ and $j \in [d]$. 
Let $M : (\{0,1\}^d)^n \to \{0,1\}^d$ be $(1,1/8 n d)$-differentially private.  Suppose that $M$ is such that for every $j$,
\begin{itemize}
\item (Assumption 1)~~$P^j \leq \frac78 - \frac{3}{16} \Longrightarrow \pr{X,M}{M(X)^j = 1} \leq \frac{ k } {16 d}$, and 
\item (Assumption 2)~~$P^j \geq \frac78 \Longrightarrow \pr{X,M}{M(X)^j = 1} \geq 1 - \frac{1}{16}$,
\end{itemize}
then $n \geq \frac{1}{16} \sqrt{k} \log(\frac{d}{16k})$.
\end{cor}

The assumptions of the theorem may seem a bit confusing, so we will clarify a bit.  Note that for every $j \in [d]$ we have $\pr{}{P^j > \frac78} \geq \frac{2k}{d}$ (Lemma~\ref{lem:BetaL}), so in expectation there are at least $2k$ such values $P^j$.  
Thus, the second assumption implies that on average $M(X)^j$ must have at least $\frac{30}{16}k$ non-zero entries.  Thus, the parameter $k$ plays roughly the same role in this problem as it does for the top-$k$ selection problem.  Furthermore, since $P^1, \cdots, P^d$ are independent, $\left|\left\{ j \in [d] ~:~ P^j > \frac78 \right\}\right|$ concentrates around its expectation.
%After proving this corollary we will elaborate further on this result and its connection to multiple hypothesis testing.

\begin{proof}[Proof of Corollary \ref{cor:mhtlb}]
First, we can lower bound the expected norm of $M(X)$ by
\begin{align*}
\ex{P,X,M}{\| M(X) \|_2^2 } 
&= \ex{P,X,M}{ \sum_{j \in [d]} M(X)^j }  \\
&= \sum_{j \in [d]} \pr{P,X,M}{M(X)^j=1} \tag{$M(X)^j \in \zo$}\\
&\geq \sum_{j \in [d]} \pr{P}{P^j \geq \frac78} \cdot \pr{P,X,M}{M(X)^j =1~\mid~ P^j \geq \frac78} \\
&\geq \sum_{j \in [d]} \frac{2k}{d} \cdot \frac{15}{16} \tag{Lemma~\ref{lem:BetaL} and Assumption 2} \\
&\geq k
\end{align*}

We need to relate this quantity to $\ex{}{\sum_j M(X)^j (P^j - \frac12)}$.  As a shorthand, let $\tau := \frac78 - \frac{3}{16} = \frac{11}{16}$ be the constant from assumption 2.
We start by writing
\begin{align*}
\ex{P,X,M}{ \sum_{j \in [d]} M(X)^j \left( P^j -\frac12 \right)}
= \sum_{j \in [d]} \ex{P,X,M}{ M(X)^j \left(P^j - \frac12 \right)}
= \sum_{j \in [d]} (A^j) + (B^j)
\end{align*}
where we define
\begin{align*}
(A^j) &:= \pr{P}{P^j \leq \tau} \cdot \ex{P,X,M}{ \left. M(X)^j \left(P^j - \frac12\right) ~\right|~ P^j \leq \tau} \\
(B^j) &:= \pr{P}{P^j > \tau} \cdot \ex{P,X,M}{ \left. M(X)^j \left(P^j - \frac12\right) ~\right|~P^j > \tau}
\end{align*}
We will manipulate each of the three terms separately.  First, for $(A^j)$, using our first assumption on $M$ we can calculate
\begin{align*}
(A^j) &= \pr{P}{P^j \leq \tau} \cdot \ex{P,X,M}{ \left. M(X)^j \left(P^j - \frac12\right) ~\right|~ P^j \leq \tau} \\
&\geq \pr{P}{P^j \leq \tau} \cdot \ex{P,X,M}{ \left. M(X)^j ~\right|~ P^j \leq \tau}  \cdot\left(0-\frac{1}{2}\right) \tag{$M(X)^j \geq 0$} \\
&= \pr{P}{P^j \leq \tau} \cdot \left(  \ex{P,X,M}{ \left. M(X)^j ~\right|~ P^j \leq \tau} \cdot \left(\tau - \frac12\right)- \tau \cdot \ex{P,X,M}{ \left. M(X)^j ~\right|~ P^j \leq \tau} \right) \\
&\geq \pr{P}{P^j \leq \tau} \cdot \left(  \ex{P,X,M}{ \left. M(X)^j ~\right|~ P^j \leq \tau} \cdot \left(\tau - \frac12\right) - \tau\cdot \frac{k}{16 d} \right) \tag{Assumption 1} \\
&\geq \pr{P}{P^j \leq \tau} \cdot  \ex{P,X,M}{ \left. M(X)^j ~\right|~ P^j \leq \tau} \cdot \left(\tau - \frac12\right) - \frac{k}{16 d} \tag{$\tau \leq 1$}
\end{align*}
And, for $(B^j)$, we can calculate
\begin{align*}
(B^j) &= \pr{P}{P^j > \tau} \cdot \ex{P,X,M}{ \left. M(X)^j \left(P^j - \frac12\right) ~\right|~P^j > \tau} \\
&\geq \pr{P}{P^j > \tau} \cdot \ex{P,X,M}{ \left. M(X)^j ~\right|~P^j > \tau} \cdot \left( \tau - \frac12 \right) \\
\end{align*}
Combining our inequalities for $(A^j)$ and $(B^j)$ we have
\begin{align*}
&\sum_{j \in [d]} (A^j) + (B^j)\\
\geq{} & \cdot \sum_{j \in [d]} \left(\left( \tau - \frac12 \right)\ex{P,X,M}{M(X)^j} - \frac{k}{16 d} \right)\\
={} &\left( \tau - \frac12 \right) \cdot \ex{P,X,M}{ \| M(X) \|_2^2} - \frac{k}{16} \tag{$M(X)^j \in \{0,1\}$}\\
\geq{} &\left( \tau - \frac12 - \frac{1}{16} \right) \cdot \ex{P,X,M}{ \| M(X) \|_2^2} \tag{$\ex{}{\|M(X)\|_2^2} \geq k$} \\
={} &\left( \frac{11}{16} - \frac12 - \frac{1}{16} \right) \cdot \ex{P,X,M}{ \| M(X) \|_2^2} = \frac18 \cdot \ex{P,X,M}{ \| M(X) \|_2^2}  \tag{$\tau = \frac{11}{16}$} 
\end{align*}
Applying Theorem \ref{thm:general} completes the proof.
\end{proof}

\jnote{Commented out the acknowledgment to Adam for the submission.}
\section*{Acknowledgements}
We thank Adam Smith for his instrumental role in the early stages of this research.

\bibliographystyle{alpha}
\bibliography{refs}

\appendix
\section{Releasing the Dataset Mean} \label{sec:meanlb}
To illustrate the versatility of Theorem \ref{thm:general}, we show how it implies known lower bounds for releasing the mean of the dataset~\cite{BunUV14, SteinkeU15b, DworkSSUV15}.

\begin{cor}
Let $M : (\{0,1\}^d)^n \to [0,1]^d$ be $(1,1/10n)$-differentially private. Let$P^1, \cdots, P^j$ be independent draws from the uniform distribution on $[0,1]$ and let $X \in (\{0,1\}^d)^n$ be such that each $X_i^j$ is independent (conditioned on $P$) and $\ex{}{X_i^j}=P^j$ for all $i \in [n]$ and $j \in [d]$. 
 Assume $\ex{P,X,M}{\|M(X)-P\|_2^2} \leq \alpha^2 d$. If $\alpha \leq 1/18$, then $n \geq \sqrt{d}/5$.
\end{cor}
Note that we can use empirical values $\overline{X}=\frac{1}{n} \sum_{i \in [n]} X_i$ instead of population values $P$, as we have $\ex{P,X}{\|\overline{X}-P\|_2^2} = \frac{1}{n} \sum_{j \in [d]} \ex{P}{P^j(1-P^j)} \leq \frac{d}{4n}$.  In this case the accuracy assumption would be replaced with $\ex{P,X,M}{\|M(X)-\overline{X}\|_2^2} \leq (\alpha^2 - \frac{1}{4n}) d$
\begin{proof}
Let $k=\ex{P,X,M}{\|M(X)\|_2^2}$.
We have 
\begin{align*}
| k - d/3 |&= \left| \ex{P,X,M}{\|M(X)\|_2^2 - \|P\|_2^2} \right|\\
&= \left| \ex{P,X,M}{\left( \|M(X)\|_2 - \|P\|_2 \right)\left( \|M(X)\|_2 + \|P\|_2 \right)} \right|\\
&\leq \left| \ex{P,X,M}{ \|M(X)-P\|_2 \cdot 2\sqrt{d}} \right|\\
&\leq 2\sqrt{d \ex{P,X,M}{\|M(X)-P\|_2^2}}\\
&\leq 2 \alpha d.
\end{align*}
So $d(1/3-2\alpha) \leq k \leq d(1/3+2\alpha)$. Furthermore,
\begin{align*}
\ex{P,X,M}{\sum_{j \in [d]} M(X)^j \cdot \left(P^j - \frac12\right)} &= \ex{P,X,M}{\sum_{j \in [d]} P^j \cdot \left(P^j - \frac12\right)} - \ex{P,X,M}{\sum_{j \in [d]} (P^j-M(X)^j) \cdot \left(P^j - \frac12\right)}\\
\left(\text{Cauchy-Schwartz}\right)&\geq \frac{d}{4} - \sqrt{\ex{P,X,M}{\sum_{j \in [d]} (P^j-M(X)^j)^2}\ex{P,X,M}{\sum_{j \in [d]} \left(P^j-\frac12\right)^2}}\\
&\geq \frac{d}{4} -\frac{\alpha d}{2}  = d(1/4-\alpha/2) \geq \frac12 d(1/3+2\alpha) \geq \frac12 k,
\end{align*}
as long as $\alpha \leq 1/18$. Hence, if $\alpha \leq 1/18$, by Theorem \ref{thm:general} (with $\beta=1$ and $\gamma=1/2$), we have $$n \geq \frac12  \sqrt{k} \geq \frac{\sqrt{d/3-2\alpha d}}{2} \geq \frac{\sqrt{\frac{2}{9} d}}{2} \geq \frac{\sqrt{d}}{5}.$$
This completes the proof.
\end{proof}

\end{document}